\documentclass[conference,a4paper]{IEEEtran}

\usepackage{pgfplots}
\pgfplotsset{compat=newest}
\usepackage{tikz}
\usetikzlibrary{arrows,matrix,positioning,positioning,calc}
\usepackage[utf8]{inputenc}

\usepackage[innermargin=0.545in,outermargin=0.65in,top=0.545in,bottom=.7in]{geometry}
\renewcommand{\baselinestretch}{0.973}

\usepackage[english]{babel}
\usepackage[T1]{fontenc}
\usepackage{epsfig}
\usepackage{amsmath, amssymb, amsbsy}
\usepackage{mathdots}
\usepackage{xspace}
\usepackage[noend]{algpseudocode}
\usepackage[linesnumbered,ruled,vlined,titlenumbered]{algorithm2e}
\usepackage{algorithmicx}
\usepackage{color}
\usepackage{cite}
\usepackage{booktabs}
\usepackage{verbatim}
\usepackage{lipsum}
\usepackage{enumitem}
\usepackage{colortbl}
\usepackage{amsthm}
\usepackage{nicefrac}
\usepackage{flushend}

\makeatletter
\newcommand\footnoteref[1]{\protected@xdef\@thefnmark{\ref{#1}}\@footnotemark}
\makeatother

\newtheorem{theorem}{Theorem}
\newtheorem{lemma}[theorem]{Lemma}
\newtheorem{corollary}[theorem]{Corollary}

\newtheorem{example}[theorem]{Example}
\newtheorem{remark}[theorem]{Remark}

\newtheorem{definition}{Definition}

\newenvironment{mymatrix}{\begin{bmatrix}} {\end{bmatrix} }

\def\ve#1{{\mathchoice{\mbox{\boldmath$\displaystyle #1$}}%
              {\mbox{\boldmath$\textstyle #1$}}%
              {\mbox{\boldmath$\scriptstyle #1$}}%
              {\mbox{\boldmath$\scriptscriptstyle #1$}}}}

\newif\ifcomment
\commenttrue

\newcommand{\Fq}{\ensuremath{\mathbb{F}_q}}

\DeclareMathOperator{\rank}{rk}

\newcommand{\Code}{\mathcal{C}}
\newcommand{\code}{\mathcal{C}}

\renewcommand{\H}{\ve{H}}

\renewcommand{\c}{\ve{c}}

\newcommand{\G}{\ve{G}}

\newcommand{\E}{\ve{E}}

\newcommand{\Eset}{\mathcal{E}}
\newcommand{\C}{\ve{C}}
\newcommand{\R}{\ve{R}}
\newcommand{\Sset}{\mathbb{S}}
\newcommand{\Scal}{\mathcal{S}}
\newcommand{\Rset}{\mathcal{R}}
\newcommand{\Iset}{\mathcal{I}}

\newcommand{\numgroups}{\mu}

\DeclareMathOperator{\supp}{supp}

\newcommand{\Pfailure}{P_\mathrm{failure}}
\newcommand{\Pmiscorrection}{P_\mathrm{miscor.}}

\newcommand{\NN}{\mathbb{N}}
\newcommand{\IC}{\mathcal{IC}}

\makeatletter
\newcommand{\removelatexerror}{\let\@latex@error\@gobble}
\makeatother

\IEEEoverridecommandlockouts
\flushend

\begin{document}

\title{On Error Decoding of Locally Repairable and Partial MDS Codes}
\author{\IEEEauthorblockN{Lukas Holzbaur\thanks{The work of L. Holzbaur was supported by the Technical University of Munich -- Institute for Advanced Study, funded by the German Excellence Initiative and European Union 7th Framework Programme under Grant Agreement No. 291763 and the German Research Foundation (Deutsche Forschungsgemeinschaft, DFG) under Grant No. WA3907/1-1. \newline The work of Sven Puchinger has received funding from the German Israeli Project Cooperation (DIP) grant no.~KR3517/9-1.}, Sven Puchinger, Antonia Wachter-Zeh}
\IEEEauthorblockA{
  Institute for Communications Engineering, Technical University of Munich (TUM), Germany\\ Email: \{lukas.holzbaur, sven.puchinger, antonia.wachter-zeh\}@tum.de}}

\maketitle

\begin{abstract}
  We consider error decoding of locally repairable codes (LRC) and partial MDS (PMDS) codes through interleaved decoding. For a specific class of LRCs we investigate the success probability of interleaved decoding. For PMDS codes we show that there is a wide range of parameters for which interleaved decoding can increase their decoding radius beyond the minimum distance with the probability of successful decoding approaching $1$, when the code length goes to infinity.
\end{abstract}

\begin{IEEEkeywords}
  Locally Repairable Codes, Partial MDS codes, Interleaved Decoding, Metzner-Kapturowski
\end{IEEEkeywords}

\section{Introduction}

Vast growth in the popularity of cloud storage and other cloud services in recent years has led to an increased interest in coding solutions for distributed data storage. Different approaches have been considered to reduce the complexity of repairing such systems in the case of node failures, with special attention to the more likely event of a single or very few node failures, in which case an efficient repair is especially important. The most prominent approaches addressing this issue are \emph{locally repairable codes} (LRC) \cite{Gopalan2012,Kamath2014,Tamo2014,Silberstein2015,Holzbaur2018}, which limit the number of nodes involved in the repair, and \emph{regenerating codes} \cite{Rashmi2012,Tamo2017}, which aim to decrease the network traffic required for repair.
The main motivation of storage codes such as LRCs and regenerating codes is erasure correction, as these occur naturally in distributed storage systems whenever nodes fail, e.g., due to hardware failures, power outages or maintenance. As discussed, e.g., in \cite{Blaum2013}, it is often assumed that errors are detected, e.g., by a cyclic redundancy check (CRC), thereby turning into erasures. While this is likely for some causes of errors, such as, e.g., faulty sectors on a hard-drive or SSD, errors caused by faulty synchronization or bad links between the nodes cannot be detected on these lower levels. These events result in \emph{errors}, i.e., the position of occurrence is unknown, which is what we consider in this work. In particular, we show that interleaved codes, i.e., the direct sum of codes with errors occurring in the same positions, can increase the tolerance against errors. As such events are less likely, they are usually not the primary design goal for such systems. Therefore, we stress that the following results show that interleaved decoding can improve the resistance against such error events \emph{without} requiring changing the storage code, but instead can be applied on the same infrastructure. In fact, in distributed data storage the assumption of burst errors, i.e., errors that corrupt the same positions in many codewords, which is required for a possible increase of the decoding radius through interleaved decoding, is very natural (cf.~Figure~\ref{fig:illustration}). Typically, a distributed storage system stores many codewords of the storage code, where each node stores one symbol of each codeword. Hence, if, e.g., one of the servers is not synchronized correctly, all codewords will be corrupted in the same position and in this case interleaved decoding can correct more errors compared to bounded minimum distance decoding.

There have been several works that consider error correction from storage codes such as LRCs and regenerating codes. In \cite{Silberstein2015} the authors consider a concatenated structure with LRCs or regenerating codes as inner codes and rank metric codes as outer codes to protect against adversaries of different types. Regenerating codes with an error tolerance in the repair process were considered in \cite{Rashmi2012}.
In \cite{Dikaliotis2010} a hashing scheme is proposed to detect errors and protect against adversarial nodes. Efficient repair of nodes by error correction from parts of the received word was considered in \cite{Tamo2017}.
In \cite{Holzbaur2018} it was shown that the error correction radius of LRCs can be increased beyond the Johnson radius.

In this work, we first investigate interleaved decoding of interleaved LRCs that are subcodes of efficiently decodable algebraic codes and derive a bound on the success probability based on existing results for the supercodes, cf.~Section~\ref{sec:ILRC}.  In the second part of the paper, Section~\ref{sec:PMDS}, we show that we can decode some interleaved partial maximum distance separable (PMDS) \cite{Blaum2013} codes (also referred to as maximally recoverable codes \cite{Gopalan2013}), beyond their minimum distance with high probability by the decoding algorithm for high-order interleaved codes by Metzner and Kapturowski \cite{metzner1990general}.

\begin{figure}
  \begin{center}
    
    \resizebox{\linewidth}{!}{\def\x{0.5}

\begin{tikzpicture}

\node (S1) at (0,0) [draw,thick,minimum width=\x*1cm,minimum height=\x*5cm,] {};
\node (S2)  [right=\x*0.3cm of S1, draw,thick,minimum width=\x*1cm,minimum height=\x*5cm] {};
\node (S3)  [right=\x*0.3cm of S2, draw,thick,minimum width=\x*1cm,minimum height=\x*5cm] {};
\node (S4)  [right=\x*0.3cm of S3, draw,thick,minimum width=\x*1cm,minimum height=\x*5cm] {};


\node (S5)  [right=\x*0.7cm of S4, draw,thick,minimum width=\x*1cm,minimum height=\x*5cm] {};
\node (S6)  [right=\x*0.3cm of S5, draw,thick,minimum width=\x*1cm,minimum height=\x*5cm] {};
\node (S7)  [right=\x*0.3cm of S6, draw,thick,minimum width=\x*1cm,minimum height=\x*5cm] {};
\node (S8)  [right=\x*0.3cm of S7, draw,thick,minimum width=\x*1cm,minimum height=\x*5cm] {};


\node (S9)  [right=\x*0.7cm of S8, draw,thick,minimum width=\x*1cm,minimum height=\x*5cm] {};
\node (S10)  [right=\x*0.3cm of S9, draw,thick,minimum width=\x*1cm,minimum height=\x*5cm] {};
\node (S11)  [right=\x*0.3cm of S10, draw,thick,minimum width=\x*1cm,minimum height=\x*5cm] {};
\node (S12)  [right=\x*0.3cm of S11, draw,thick,minimum width=\x*1cm,minimum height=\x*5cm] {};


\draw[dotted,thick] (\x*4.9, \x*3) -- (\x*4.9,-\x*2.7);
\draw[dotted,thick] (\x*10.8, \x*3) -- (\x*10.8,-\x*2.7);


\foreach \i in {1,...,12}{
  \node (C\i) at ($(S\i)+(0,\x*2)$) [draw,minimum width=\x*0.7cm,minimum height=\x*0.5cm,rounded corners=1pt] {};
  \node (C\i) at ($(S\i)+(0,\x*1.4)$) [draw,minimum width=\x*0.7cm,minimum height=\x*0.5cm,rounded corners=1pt] {};
  \node (C\i) at ($(S\i)+(0,\x*0.8)$) [draw,minimum width=\x*0.7cm,minimum height=\x*0.5cm,rounded corners=1pt] {};
  \node (C\i) at ($(S\i)+(0,\x*0.2)$) [draw,minimum width=\x*0.7cm,minimum height=\x*0.5cm,rounded corners=1pt] {};
  \node (C\i) at ($(S\i)+(0,\x*-0.4)$) [minimum width=\x*0.7cm,minimum height=\x*0.5cm,rounded corners=1pt] {$\vdots$};
  \node (C\i) at ($(S\i)+(0,\x*-1.4)$) [draw,minimum width=\x*0.7cm,minimum height=\x*0.5cm,rounded corners=1pt] {};
  \node (C\i) at ($(S\i)+(0,\x*-2)$) [draw,minimum width=\x*0.7cm,minimum height=\x*0.5cm,rounded corners=1pt] {};

}


\draw[dashed, rounded corners = 1pt, blue, thick] ($(S1)+(-\x*.7,\x*2.4)$) rectangle ($(S12)+(\x*.7,\x*1.6)$);


\node[anchor = south east] (L1) at ($(S12)+(\x*1,\x*4)$) {\footnotesize LRC codeword};
\path (L1) edge[bend left, ->]  ($(S12)+(\x*.7,\x*2.4)$) ;

\draw [decorate,decoration={brace,amplitude=5pt}]  ($(S1)+(-\x*1,-\x*2.4)$) -- ($(S1)+(-\x*1,\x*2.4)$) node [black,midway,xshift=-0.4cm,rotate=90] {\footnotesize $\ell$ independent codewords};


\draw[thick,->,red] ($(S2)+(0,\x*2)$) -- ($(S2)+(-\x*0.3,-\x*0.4)$) -- ($(S2)+ (\x*0.3,\x*0.4)$) -- ($(S2)+(0,-\x*2)$);
\draw[thick,->,red] ($(S4)+(0,\x*2)$) -- ($(S4)+(-\x*0.3,-\x*0.4)$) -- ($(S4)+ (\x*0.3,\x*0.4)$) -- ($(S4)+(0,-\x*2)$);
\draw[thick,->,red] ($(S7)+(0,\x*2)$) -- ($(S7)+(-\x*0.3,-\x*0.4)$) -- ($(S7)+ (\x*0.3,\x*0.4)$) -- ($(S7)+(0,-\x*2)$);
\draw[thick,->,red] ($(S12)+(0,\x*2)$) -- ($(S12)+(-\x*0.3,-\x*0.4)$) -- ($(S12)+ (\x*0.3,\x*0.4)$) -- ($(S12)+(0,-\x*2)$);


\node at ($(S6)+(\x*0.6,\x*4.2)$) {Servers};
\node at ($(S2)+(\x*0.6,\x*3.2)$) {Local group $1$};
\node at ($(S6)+(\x*0.6,\x*3.2)$) {Local group $2$};
\node at ($(S10)+(\x*0.6,\x*3.2)$) {Local group $3$};


\node (El) at ($(S6)+(\x*0.6,-\x*4.2)$) {\color{red} Burst Errors};
\path ($(El.north west)+(\x*0,\x*-0.2)$) edge[bend left=5, ->]  ($(S2)+(\x*0.1,-\x*2.6)$) ;
\path ($(El.north)+ (-\x*0.5,0)$) edge[bend left=5, ->]  ($(S4)+(\x*0.1,-\x*2.6)$) ;
\path ($(El.north) + (\x*0.5,0)$) edge[bend right=5, ->]  ($(S7)+(-\x*0,-\x*2.6)$) ;
\path ($(El.north east)+(\x*0,\x*-0.2)$) edge[bend right=5, ->]  ($(S12)+(-\x*0.1,-\x*2.6)$) ;

\end{tikzpicture}

}

\end{center}
\caption{Illustration of LRC coded storage system with burst errors}
\label{fig:illustration}
\vspace{-10pt}
\end{figure}
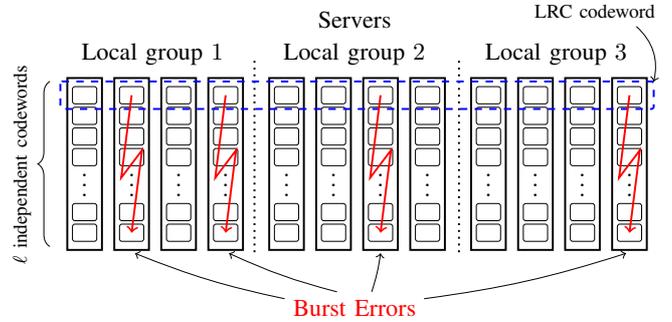

\section{Preliminaries}

\subsection{Locally Repairable Codes}

A code is said to have locality $r$ if every position can be recovered from at most $r$ other codeword positions. If multiple erasures can be tolerated within such a local \emph{repair set}, the code is said to have $(r,\rho)$ locality.

\begin{definition}[$(r,\rho)$-locality]
  A $[n,k]$ code $\code$ has $(r,\rho)$-locality if there exists a partition $\mathcal{A} = \{A_1,A_2,...,A_\mu\}$ into sets of cardinality $|A_j| \leq r+\rho-1$ such that $\forall i \in [n], \ \exists j\in [\mu]$ s.t. $i \in A_j$ and $d\left(\code |_{A_j}\right) \geq \rho$, where $\code |_{A_j}$ denotes the restriction of the code to the positions indexed by $A_j$.
\end{definition}

A Singleton-like upper bound on the achievable distance of an $[n,k,r,\rho]$ LRC was derived in \cite{Gopalan2012} for $\rho=2$ and generalized to $\rho \geq 2$ in \cite{Kamath2014} to
\begin{equation} \label{eq:boundDistanceLRC}
  d \leq n-k+1-\left( \left\lceil \frac{k}{r} \right\rceil -1 \right) (\rho-1) \ .
\end{equation}
In the following we refer to codes achieving this bound with equality as \emph{optimal} LRCs.

Several classes of optimal LRCs are known \cite{tamo2016optimal, Silberstein2013}, including one of particular interest for this work, the so-called Tamo-Barg LRCs \cite{Tamo2014}. The advantage of the latter is, that since they are constructed as subcodes of RS codes, the requirement on the field size is only $q\geq n$ and they can be decoded by any of the well-studied RS decoders. 

\subsection{Partial MDS Codes}

In Section~\ref{sec:PMDS} we consider interleaved decoding of a special class codes with locality, namely partial MDS codes \cite{Blaum2013,gabrys2018constructions,Blaum2016}, also referred to as \emph{maximally recoverable codes} \cite{Huang2007, Chen2007, Gopalan2013}. The distinctive property of these codes is that they guarantee to correct any pattern of erasures that is information theoretically correctable.
\begin{definition}[PMDS codes]\label{def:partialMDS}
   Let $\Code$ be an $[n,k,r,\varrho]$ LRC code, where the local groups are $[r+\varrho-1,r,\varrho]$ MDS codes. We say that the code $\code$ is a \emph{partial MDS code} if for any set $E \subset [n]$, where $E$ is obtained by picking $\varrho-1$ positions from each local group, the distance of the code punctured in these positions is $d(\code|_E) = n-\frac{n(\varrho-1)}{r+\varrho-1}-k+1$.
 \end{definition}
 \begin{remark}
   A more general form of PMDS codes, where the distance, i.e., number of tolerable erasures, can be different in each local group is often considered in literature. For simplicity, we focus on PMDS with the same distance in each local group in this work.
 \end{remark}

\subsection{Interleaved Codes}

Interleaved codes are direct sums of a number of constituent codes.
We will only consider homogeneous interleaved codes over linear codes in this work, for which the constituent codes are all the same and linear.

\begin{definition}[See, e.g., \cite{metzner1990general,krachkovsky1997decoding}]
Let $\Code[n,k,d]$ be a linear code over $\Fq$ and $\ell \in \NN$. The corresponding $\ell$-interleaved code is defined by
\begin{align*}
\IC[\ell; n,k,d] := \left\{ \C = \left[\begin{smallmatrix}
      \c_1 \\ \c_2 \\ \vphantom{\int\limits^x}\smash{\vdots} \\ \c_\ell
\end{smallmatrix}\right] \, : \, \c_i \in \Code \right\}.
\end{align*}
\end{definition}

The assumed error model is as follows.
We want to reconstruct a codeword $\C$ from a received word of the form $\R = \C + \E$, where $\E \in \Fq^{\ell \times n}$ is an error matrix. Let $\Eset$ be the set of indices of non-zero columns of $\E$, then we say that an error matrix is of weight $t$ if $|\Eset| = t$.

For some constituent codes, for instance RS or some AG codes, there are efficient decoders that correct many errors beyond half the minimum distance and even the Johnson radius with high probability.
The first such algorithm was given in \cite{krachkovsky1997decoding} for RS codes and corrects up to $\tfrac{\ell}{\ell+1}(n-k)$ errors.
Since then, many decoders with better complexity and larger decoding radius, as well as some bounds on the probability of decoding failure have been derived.
Due to space restrictions, we cannot list all of the papers. One decoder of special interest for this work was introduced by Metzner and Kapturowski in \cite{metzner1990general} and will be discussed in more detail in Section~\ref{sec:metznerKapturowski}.

\section{Interleaved Locally Repairable Codes}
\label{sec:ILRC}

In data storage, as in data transmission, codes over small fields are generally favorable as they allow for lower complexity decoding of errors or recovery from erasures. The advantage of interleaved codes is that in most cases, i.e., when $<d/2$ errors occur, it is sufficient to consider each codeword separately, thereby keeping the decoding complexity low. Only in a worst case scenario where $\geq d/2$ errors occur, the stored codewords can be viewed as an interleaved code, hence increasing the decoding radius and resolving the errors with high probability. 

\subsection{Decoding in the Interleaved Supercode} \label{subsec:decSupercode}

We first consider the class of Tamo--Barg LRCs \cite{Tamo2014}, which has attracted a lot of attention in recent years. These codes are subcodes of RS codes of the same minimum distance (we refer to the RS code that contains the LRC as its supercode) and can therefore be decoded by any RS decoding algorithm.

Due to the subcode property, we can also directly apply any interleaved RS decoder to decode interleaved Tamo--Barg LRCs.
In the following, we will recall the properties of the decoder in \cite{schmidt2009collaborative}, which is a unique decoder with maximal decoding radius
\begin{align}
  t \leq t_\mathrm{max} :=  \frac{\ell}{\ell+1}(d-1) \ . \label{eq:tmax}
\end{align}
This means that it returns one of the following possible outputs:
\begin{itemize}
\item the transmitted codeword (\emph{success}),
\item another codeword (\emph{miscorrection}/\emph{undetected error event}),
\item or no codeword (\emph{failure}).
\end{itemize}
Up to half the minimum distance, the decoder always succeeds. 
Beyond $d/2$, bounds on the probability of miscorrection and failure, assuming that we choose an error of weight $t$ uniformly at random, are given in \cite{schmidt2009collaborative}.
These probabilities are shown to decrease exponentially in the difference between the maximum decoding radius $t_\mathrm{max}$ and the number of errors $t$, where generally $\Pmiscorrection \ll \Pfailure$. Furthermore, the authors of \cite{schmidt2009collaborative} showed that decoding success, failures, and miscorrections only depend on the error matrix and are independent of the codeword.

Since Tamo--Barg codes are subcodes of RS codes with the same minimum distance, we can decode interleaved Tamo--Barg codes with the decoder in \cite{schmidt2009collaborative}.
If the returned codeword is not a codeword of the interleaved Tamo--Barg code, then we can declare this as a decoding failure.
Since this only happens in case of a miscorrection, the sum of the number of miscorrections and failure events is unchanged. This proves the following corollary.

\begin{corollary}\label{col:TBbound}
  $\ell$-Interleaved Tamo--Barg codes of minimum distance $d$ can correct up to $t_{\mathrm{max}}$ errors and decoding succeeds for a fraction $1-\Pfailure-\Pmiscorrection$, where $\Pfailure$, and $\Pmiscorrection$ are defined as in \cite{schmidt2009collaborative} and $t_{\mathrm{max}}$ is given in (\ref{eq:tmax}).
\end{corollary}

It is well known that even for small interleaving orders and short code lengths interleaved decoding is successful with high probability. In a real storage system these values will likely be very high, as the interleaving order is naturally very large. For example, consider a $[n=15,k=8,r=4,\rho=2]$ storage code of distance $d=7$ operating on bytes, i.e., over the field $F_{2^8}$. The unique decoding radius of this code is $\left\lfloor \frac{d-1}{2} \right\rfloor = 3$. Now assume burst errors occurring on hard-drive sectors of typical size $512$ bytes. This results in an interleaving order of $\ell = 512$ and an interleaved decoding radius of $t = 5$. The bound of Corollary~\ref{col:TBbound} gives a success probability $> 1- 10^{-1223}$.

\begin{remark}
 Note that these results only hold for random errors, not adversarial errors, as the probability of successful decoding is related to the rank of the error matrix.
\end{remark}

\section{Decoding of Partial MDS codes beyond their Minimum Distance}
\label{sec:PMDS}
The decoding technique of Section~\ref{subsec:decSupercode} is not specific to LRCs, as the discussed effects on the probability of the decoder outcome always hold when decoding a code by a decoder of its supercode.
In this section we show that if the necessary conditions on the rate are fulfilled, the Metzner-Kapturowski decoding algorithm \cite{metzner1990general} corrects up to $t=n-k-1$ errors in interleaved partial MDS code with probability going to one when the code length goes to infinity. Besides the increase in the decoding radius, another advantage of this decoder is that it is a generic decoder, i.e., it can be applied to any code without requiring any additional structure.

\subsection{A Generalization of Metzner and Kapturowski's Statement} \label{sec:metznerKapturowski}

Metzner and Kapturowski proved in \cite[Theorem~2]{metzner1990general} that a codeword $\C$ of an interleaved code with minimum distance $d$ can be uniquely recovered from a corrupted word $\C+\E$ if
\begin{enumerate}
\item The number of errors is $t := |\Eset| \leq d-2$ and
\item the error matrix $\E$ has full rank $t$ (this implicitly assumes that the interleaving order is high, i.e., $\ell \geq t$),
\end{enumerate}
However, the first condition is very restrictive and not necessary for the decoder to work.
In fact, the proof of \cite[Theorem~2]{metzner1990general} only assumes an implication of the first property: The $t+1$ columns of the parity-check matrix indexed by the error positions $\Eset$ and any other integer in $\{1,\dots,n\} \setminus \Eset$ must be linearly independent.
We will give this property a name in the following definition.

\begin{definition}
Let $\H \in \Fq^{n-k \times n}$ be a parity-check matrix of a linear code $\Code[n,k,d]$.
A set $\Eset \subseteq \{1,\dots,n\}$ with $t = |\Eset|$ is called \emph{$(t+1)$-independent (with respect to $\H$)} if
\begin{align*}
\rank \left(\H_{\Eset \cup \{i\}}\right) = t+1	\quad \forall i \in \{1,\dots,n\} \setminus \Eset,
\end{align*}
where $\H_{\Eset \cup \{i\}}$ is the matrix consisting of the columns of $\H$ indexed by $\Eset \cup \{i\}$.
\end{definition}

Note that for $t\leq d-2$, any set is $t+1$-independent, and for $t\geq n-k$, no set is $t+1$-independent.
Using this definition, we can formally state a generalization of \cite[Theorem~2]{metzner1990general}.

\begin{theorem}\label{thm:MK_generalization}
Let $\C \in \IC$ be a codeword of an interleaved code and let $\H$ be a parity-check matrix of its constituent code $\Code$.
Let $\E$ be an error matrix with $t$ non-zero columns, indexed by $\Eset$.
Then $\C$ can be uniquely recovered from $\C+\E$ by the Metzner--Kapturowski algorithm \cite{metzner1990general} if $\Eset$ is $t+1$-independent w.r.t.\ $\H$ and $\rank\left(\E_\Eset\right)=t$.
\end{theorem}

Note that the second condition, $\rank\left(\E_\Eset\right)=t$, is fulfilled for most error matrices with $t$ non-zero columns if the interleaving order $\ell$ is large enough.

In the following subsections, we will see that Theorem~\ref{thm:MK_generalization} is indeed an improvement over \cite[Theorem~2]{metzner1990general} since there are codes with only a few error positions $\Eset$ that are not $t+1$-independent for $t>d-2$.

\subsection{PMDS Codes With Many $(t+1)$-Independent Positions}

A set of erasures $\Eset$ can be corrected if and only if its complement $\bar{\Eset} := \{1,\dots,n\} \setminus \Eset$ contains an \emph{information set}, i.e., indexes $k$ linearly independent columns of the generator matrix. 
The authors of \cite{tamo2016optimal} studied a family of optimal LRCs, which in some parameter range are able to correct $n-k$ erasures with probability approaching $1$ for large code lengths. This follows from showing that the number of information sets relative to the number of all sets with $k$ elements tends to $1$ for $n \to \infty$.

We will use a similar approach in the following to show that the relative number of $(t+1)$-independent positions with $t \leq n-k-1$ tends to $1$ for a family of PMDS codes.

\begin{definition}[\!\!\cite{tamo2016optimal}]
Let $\Rset_1,\dots,\Rset_{n/(r+\varrho-1)}$ be the repair sets of an $[n,k,r,\varrho]$ PMDS code. We define the set
\begin{equation*}
\Sset_{\mu} := \{ \Scal \subseteq \{1,\dots,n\} \, : \, |\Scal|=\mu , \, |\Scal \cap \mathcal{R}_i| \leq r, \, \forall i  \}.
\end{equation*}
\end{definition}

The following was shown in \cite{tamo2016optimal} for a special class of PMDS codes and holds in general for PMDS codes.

\begin{lemma}\label{lem:information_sets_Sk}
Let $\G$ be a generator matrix of an $[n,k,r,\delta]$ PMDS code.
Let $\Scal \subseteq \{1,\dots,n\}$ be of cardinality $k$.
Then, the columns of $\G_\Scal$ (i.e., the columns of $\G$ indexed by $\Scal$) are linearly independent if and only if $\Scal \in \Sset_k$.
\end{lemma}

\begin{proof}
The statement was proven for the codes in \cite{tamo2016optimal} within the proof of \cite[Lemma~7]{tamo2016optimal}. It holds in general for PMDS codes since any set $\Scal \in \Sset_k$ corresponds to $k$ columns of a generator matrix of the MDS code obtained by puncturing the PMDS code at $\varrho-1$ positions of each local group not in $\Scal$. This puncturing is possible since $\Scal$ intersects with each local group in at most $r$ positions, so there are at least $\varrho-1$ positions left in each group. It is well-known that any $k$ columns of an MDS code's generator matrix are linearly independent.
\end{proof}

The following lemma relates the sets in $\Sset_{\mu}$ to the $(t+1)$-independent property.

\begin{lemma}\label{lem:t+1-independent_S_k+1}
Let $\H \in \Fq^{n-k \times n}$ be a parity-check matrix of an $[n,k,r,\delta]$ partial MDS code.
Then, a set $\Iset$ of cardinality $t := |\Iset|$ is $(t+1)$-independent if and only if it is a subset of the complement of an element $\Scal \in \Sset_{k+1}$, i.e., 
\begin{equation*}
\Iset \subseteq \bar{\Scal} := \{1,\dots,n\} \setminus \Scal. 
\end{equation*}
\end{lemma}

\begin{proof}
First note that $\Scal' \in \Sset_k$ if and only if the complementary columns of $\H$, i.e., $\H_{\bar{\Scal'}}$, have full rank $n-k$. This is due to Lemma~\ref{lem:information_sets_Sk} and the fact that the columns of a generator matrix $\G$ indexed by $\Scal'$ have full rank. Thus, we can find a quasi-systematic parity-check matrix with the identity matrix in the complementary columns (i.e., the complement of an information set of a code is an information set of its dual code), and vice-versa.

The set $\Iset$ is $(t+1)$-independent if and only if the columns $\Iset \cup \{i\}$ of $\H$ are linearly independent for any additional column~$i$.
This again is true iff $\Iset \cup \{i\}$ is contained in an information set of the dual code.
By the above argument, this is equivalent to $\Iset \cup \{i\}$ being contained in the complement of some $\Scal_i' \in \Sset_k$.
Since $i$ is arbitrary, this holds iff $\Iset$ is in the complement of a set $\Scal$ of cardinality $k+1$, that, if any one element is removed, is in $\Sset_k$. It follows that $\Scal$ must be in $\Sset_{k+1}$.
\end{proof}

Due to Lemma~\ref{lem:t+1-independent_S_k+1}, the relative amount of $t+1$-independent positions can be lower-bounded using the set $\Sset_{k+1}$ as follows.

\begin{lemma}
Let $t \leq n-k-1$. Then,
\begin{equation*}
\frac{|\{ \Iset \subseteq \{1,\dots,n\} \, : \, |\Iset| = t, \, \text{$(t+1)$-independent} \}|}{|\{ \Iset \subseteq \{1,\dots,n\} \, : \, |\Iset| = t \}|} \geq \frac{|\Sset_{k+1}|}{\binom{n}{k+1}}.
\end{equation*}
\end{lemma}

\begin{proof}
If $t=n-k-1$, then there is a one-to-one correspondence between the $(t+1)$-independent sets and the elements in $\Sset_{k+1}$, given by the complement of the respective set. Hence, the bound is fulfilled with equality. For $t<n-k+1$, a set is $t+1$-independent if and only if its complement contains an element of $\Sset_{k+1}$. The number of cardinality-$(n-t)$ sets containing at least one element of $\Sset_{k+1}$ relative to all cardinality-$(n-t)$ sets is greater or equal to $\tfrac{|\Sset_{k+1}|}{\binom{n}{k+1}}$ since the relative number of sets having no element of $\Sset_{k+1}$ as a subset decreases in $n-t$.
\end{proof}

The following theorem is a generalization of \cite[Theorem~3]{tamo2016optimal}, which lower-bounds $\nicefrac{\Sset_{k}}{\binom{n}{k}}$ for the special case $\varrho=2$.

\begin{lemma}\label{lem:Partial_MDS_bound_S_k+1}
Let $\Code$ be an $[n,k,r,\varrho]$ PMDS code.
Then,
\begin{align}
\frac{|\Sset_{k+1}|}{\binom{n}{k+1}} \geq 1-2^{\log_2(n) - (r+1) \log_2\left( \binom{r+\varrho-1}{\xi}^{-\frac{1}{r+1}} \frac{n}{k+1}\right)}, \label{eq:bound_on_S_k+1}
\end{align}
where $\xi := \min\left\{\varrho-2,\left\lfloor\frac{r+\varrho-1}{2}\right\rfloor\right\}$.
\end{lemma}

\begin{proof}
We have
\begin{align*}
&\binom{n}{k+1}-|\Sset_{k+1}| = |\overline{\Sset_{k+1}}| \\
&= \left| \left\{ \Scal \subseteq \{1,\dots,n\} \, : \, |\Scal|=k+1 , \, \exists i \, : \, |\Scal \cap \mathcal{R}_i| > r \right\} \right| \\
&\leq \sum_{i=1}^{\numgroups} \left| \left\{ \Scal \subseteq \{1,\dots,n\} \, : \, |\Scal|=k+1 , \, \, |\Scal \cap \mathcal{R}_i| > r \right\} \right| \\
&\leq \mu \sum_{j=r+1}^{r+\varrho-1} \binom{r+\varrho-1}{j} \underbrace{\binom{n-j}{k+1-j}}_{\leq (\frac{k+1}{n})^{r+1}\binom{n}{k+1} } \\
&\leq \mu \left(\frac{k+1}{n}\right)^{r+1}\binom{n}{k+1} \sum_{j=r+1}^{r+\varrho-1} \underbrace{\binom{r+\varrho-1}{j}}_{\leq \binom{r+\varrho-1}{\xi}} \\
&\leq \mu (\varrho-1) \left(\frac{k+1}{n}\right)^{r+1}\binom{n}{k+1} \binom{r+\varrho-1}{\xi}
\end{align*}
Hence, we have
\begin{align*}
\frac{|\Sset_{k+1}|}{\binom{n}{k+1}} &\geq 1-\mu (\varrho-1) \left(\binom{r+\varrho-1}{\xi}^{\frac{1}{r+1}} \frac{k+1}{n}\right)^{r+1} \\
                                     &\geq 1-n \left(\binom{r+\varrho-1}{\xi}^{\frac{1}{r+1}} \frac{k+1}{n}\right)^{r+1} \\
&= 1-2^{\log_2(n) - (r+1) \log_2\left( \binom{r+\varrho-1}{\xi}^{-\frac{1}{r+1}} \frac{n}{k+1}\right)} \ ,
\end{align*}
which proves the claim.
\end{proof}

Note that for $\varrho \leq r+2$, we always have $\xi = \varrho-2$.

Using the bound in Lemma~\ref{lem:Partial_MDS_bound_S_k+1}, we are able to formulate conditions on the local and global distance of a family of PMDS codes for which the relative size of $\Sset_{k+1}$ compared to all cardinality-$k+1$ subsets of $\{1,\dots,n\}$ approaches $1$ for growing code length.

\begin{lemma}\label{lem:Partial_MDS_bound_S_k+1_asymptotical}
Let $\{\Code_n\}$ be a familiy of $[n,k_n,r_n,\varrho_n]$ PMDS LRC with
\begin{align}
\binom{r_n+\varrho_n-1}{\xi_n}^{-\frac{1}{r_n+1}}  &> C_1 \frac{k_n+1}{n}\label{eq:rate_condition}\\
    r_n+1 &\geq \frac{C_2 \log_2(n)}{\log_2(C_1)} \label{eq:number_of_local_groups_condition}
\end{align}
for some $C_1,C_2>1$, where $\xi_n := \min\left\{\varrho_n-2,\left\lfloor\frac{r_n+\varrho_n-1}{2}\right\rfloor\right\}$.
Then,
\begin{align*}
\frac{|\Sset_{k_n+1}|}{\binom{n}{k_n+1}} \to 1 \quad (n \to \infty).
\end{align*}
\end{lemma}

\begin{proof}
It is easy to see that the exponent of $2$ in the bound \eqref{eq:bound_on_S_k+1} converges to minus infinity under the given conditions.
\end{proof}

\begin{remark}
Condition~\eqref{eq:number_of_local_groups_condition} puts a rate constraint on the code. However, if $r_n$ grows faster to infinity than $\varrho_n$, the following argument shows that we can choose arbitrary rates.
We study the asymptotic behavior of $\binom{r_n+\varrho_n-1}{\xi_n}^{-\frac{1}{r_n+1}}$ for $r_n \in \omega(\varrho_n)$ (i.e., $r_n$ grows asymptotically strictly faster than $\varrho_n$):
\begin{align*}
\binom{r_n+\varrho_n-1}{\xi_n}^{-\frac{1}{r_n+1}} &= \frac{1}{\binom{r_n+\varrho_n-1}{\varrho_n-2}^{\frac{1}{r_n+1}}} \\
&\geq \frac{1}{\left(\frac{e(r_n+\varrho_n-1)}{\varrho_n-2}\right)^{\frac{\varrho_n-2}{r_n+1}}} \\
&=\frac{1}{\underbrace{e^{\frac{\varrho_n-2}{r_n+1}}}_{\to \, 1} \cdot \underbrace{\left(1+\frac{r_n+1}{\varrho_n-2}\right)^{\frac{\varrho_n-2}{r+1}}}_{\to \, 1} } \to 1.
\end{align*}
Note that we use that if $r_n$ grows faster than $\varrho_n$, at some point we have $\xi_n = \varrho_n-2$.
\end{remark}

\begin{remark}
By a similar argument as above, Condition~\eqref{eq:rate_condition} in Lemma~\ref{lem:Partial_MDS_bound_S_k+1_asymptotical} can be replaced by
\begin{equation*}
 \left(\frac{R_\mathsf{local}}{e}\right)^{\varrho-2} \approx \left(\frac{r+1}{e(r+\varrho-1)}\right)^{\varrho-2} > C_1 \frac{k+1}{n} \approx C_1 R_\mathsf{global}
\end{equation*}
\end{remark}

Lemma~\ref{lem:Partial_MDS_bound_S_k+1_asymptotical} implies that under the given conditions, asymptotically almost any set of $k+1$ indices is in $\Sset_{k+1}$, and thus, almost any set of $t \leq n-k-1$ error positions is $(t+1)$-independent.

\begin{remark}
Since any cardinality-$k$ subset of $\Sset_{k+1}$ is an information set, Lemma~\ref{lem:Partial_MDS_bound_S_k+1_asymptotical} also implies that the codes satisfying Conditions \eqref{eq:rate_condition} and \eqref{eq:number_of_local_groups_condition}, and can correct almost all $n-k$ erasures asymptotically. This constitutes a generalization of the statement in \cite[Theorem~3]{tamo2016optimal}, which proves the special case $\varrho=2$.
\end{remark}

\subsection{Decoding PMDS Codes Beyond the Minimum Distance}

Using Lemma~\ref{lem:Partial_MDS_bound_S_k+1_asymptotical}, we can state the following explicit class of PMDS codes correcting almost any error up to weight $n-k-1$.

\begin{theorem}\label{thm:PMDS_family_correcting_n-k-1}
Let $\{\Code_n\}$ be a familiy of $[n,k_n,r_n,\varrho_n]$ PMDS LRC over a field $q_n$, where $q_n \to \infty$ for $n \to \infty$ and the parameters $r_n, \varrho_n$ fulfill
conditions \eqref{eq:rate_condition} and \eqref{eq:number_of_local_groups_condition} of Lemma~\ref{lem:Partial_MDS_bound_S_k+1_asymptotical} for fixed constants $C_1,C_2>1$.
Furthermore, let $\ell_n = n-k_n-1$. Then, the family $\{\Code_n'\}$ of $[n,k_n,r_n,\varrho_n]$ codes over the fields of size $q_n^{\ell_n}$ obtained by interpreting the $\ell_n$-interleaved codes of $\Code_n$ as linear codes over the large field $\mathbb{F}_{q_n^{\ell_n}}$, fulfill the following properties:
\begin{itemize}
\item the codes $\Code_n'$ are PMDS,
\item $\Code_n'$ corrects up to $n-k_n-1$ errors with probability approaching $1$ for $n \to \infty$ (assuming uniformly distributed errors of given weight), and
\item the decoding complexity is $O(n^3)$ operations over $\mathbb{F}_{q_n}$.
\end{itemize}
\end{theorem}

\begin{proof}
A (homogeneous) interleaved code is a linear code over the large field of the same parameters as the constituent code. Since puncturing the interleaved code corresponds to puncturing the constituent codes, the definition of PMDS codes directly implies that the interleaved code is also PMDS.

For showing the correction capability, first note that interpreting elements of $\mathbb{F}_{q_n^{\ell_n}}$ as vectors in $\mathbb{F}_{q_n}^{\ell_n}$ gives a bijective mapping between all Hamming errors of weight $t$ in $\mathbb{F}_{q_n^{\ell_n}}^n$ and all burst errors in $\mathbb{F}_{q_n}^{\ell_n \times n}$ of weight $t$.
We use Theorem~\ref{thm:MK_generalization}, Lemma~\ref{lem:Partial_MDS_bound_S_k+1_asymptotical}, and the fact that the fraction of $\ell_n \times t_n$ matrices over the field of size $q_n$ of rank $t_n$ is at least $1-\nicefrac{4}{q_n}$ for $q_n\geq 4$, cf.~\cite[Lemma~3.13]{Overbeck_Diss_InterleveadGab}.
The probability that a random error pattern of weight $t<n-k$ cannot be corrected is therefore
\begin{align*}
&\mathrm{P}(\E \text{ cannot be corrected}) \\
&=\mathrm{P}\left(\Scal \not\subseteq \overline{\supp(\E)} \, \forall \Scal \in \Sset_{k_n+1} \, \lor \,  \rank_{\Fq}(\E)<t_n \right) \\
&= \mathrm{P}\left(\Scal \not\subseteq \overline{\supp(\E)} \, \forall \Scal \in \Sset_{k_n+1}\right) + \underbrace{\mathrm{P} \left(\rank_{\Fq}(\E)<t_n \right)}_{< \, 4/q_n} \\
&< 1-\frac{|\Sset_{k_n+1}|}{\binom{n}{k_n+1}} + \frac{4}{q_n} \to 0 \quad (n \to \infty)
\end{align*}
since $q_n \to \infty$ for $n \to \infty$.

As for the complexity, we apply the Metzner--Kapturowski algorithm on the received matrix in $\mathbb{F}_{q_n}^{\ell_n \times n}$, which runs in complexity $O(n^3)$ over the small field $\mathbb{F}_{q_n}$ since $\ell_n \leq n$.
\end{proof}

\begin{remark}
If the assumption $q_n \to \infty$ for $n \to \infty$ in Theorem~\ref{thm:PMDS_family_correcting_n-k-1} is not fulfilled for a class of codes, this disproves the MDS conjecture. 
\end{remark}

The overall field size $Q_n$ for two families of PMDS codes is given in the following without proof.

\begin{corollary}
Let the family $\{\Code_n\}$ be a subset of the code class in \cite{tamo2016optimal}. Then, the field size is given by
\begin{align*}
\log Q_n &\in O\left( n^{2+\log(\log(n))}\log(n) \right).
\end{align*}
and the overall decoding complexity in bit operations is
\begin{align*}
O^\sim\!\left( n^{4+\log(\log(n))} \right),
\end{align*}
where $O^\sim$ neglects logarithmic terms in $n$.
\end{corollary}

\begin{corollary}
Let the family $\{\Code_n\}$ be a subset of the code class in \cite{gabrys2018constructions}. Then, the field size is given by
\begin{align*}
\log Q_n &\in O \left( n^3 log(log(n)) \right).
\end{align*}
and the overall decoding complexity is $O^\sim(n^5)$ bit operations.
\end{corollary}
For the special case of $\varrho = 2$ the probability of successful decoding can be stated exactly.
\begin{corollary} \label{col:PsucExact}
  The probability of successfully decoding $t$ errors in an $[n,k,r,\varrho=2]$ PMDS code is given by
  \begin{equation*}
    P_{\mathsf{suc}} =  P\{\rank (E) = t\} -\frac{|\Sset_{k+1}|}{\binom{n}{k+1}} \ ,
  \end{equation*}
  where, as shown in \cite[Proof of Theorem~3]{tamo2016optimal},
  \begin{equation*}
    |\Sset_{k+1}| = \sum_{j=1}^{\left\lfloor \frac{k+1}{r+1} \right\rfloor} (-1)^{j-1}\binom{n/(r+1)}{j} \binom{n-j(r+1)}{k+1-j(r+1)}
  \end{equation*}
  and the fraction of full rank matrices $E \in \Fq^{\ell \times t}$ \cite{migler2004} is  
  \begin{equation*}
    P\{\rank(E) = t\} = q^{-t \ell} \prod_{j=0}^{t-1} (q^\ell-q^j) \ .
  \end{equation*}

\end{corollary}

The following example shows that the success probability is reasonably close to $1$ even for small parameters.

\begin{example}
  Consider the PMDS code as defined in \cite{tamo2016optimal} with parameters $[n=15,k=8,r=4,\varrho=2]$ over the field $\Fq$ with $q=16^{k+1} = 2^{36}$. The code is of distance $d=7$, fulfilling the bound (\ref{eq:boundDistanceLRC}) on the distance of an LRC. The unique decoding radius of this code is $t=\left\lfloor \frac{d-1}{2} \right\rfloor = 3$. Given a full rank error matrix, the decoder introduced in \cite{metzner1990general} guarantees decoding of up to $t=d-2=5$ errors. In the case of $t=n-k-1=6$ errors, the error matrix is of full rank with probability $>1-10^{10}$ and Corollary~\ref{col:PsucExact} gives the probability of success as $P_{\mathsf{dec}} \approx \frac{125}{143}$.
\end{example}

\vspace{4pt}
\bibliographystyle{IEEEtran}
\bibliography{main}

\begin{thebibliography}{10}
\providecommand{\url}[1]{#1}
\csname url@samestyle\endcsname
\providecommand{\newblock}{\relax}
\providecommand{\bibinfo}[2]{#2}
\providecommand{\BIBentrySTDinterwordspacing}{\spaceskip=0pt\relax}
\providecommand{\BIBentryALTinterwordstretchfactor}{4}
\providecommand{\BIBentryALTinterwordspacing}{\spaceskip=\fontdimen2\font plus
\BIBentryALTinterwordstretchfactor\fontdimen3\font minus
  \fontdimen4\font\relax}
\providecommand{\BIBforeignlanguage}[2]{{%
\expandafter\ifx\csname l@#1\endcsname\relax
\typeout{** WARNING: IEEEtran.bst: No hyphenation pattern has been}%
\typeout{** loaded for the language `#1'. Using the pattern for}%
\typeout{** the default language instead.}%
\else
\language=\csname l@#1\endcsname
\fi
#2}}
\providecommand{\BIBdecl}{\relax}
\BIBdecl

\bibitem{Gopalan2012}
P.~Gopalan, C.~Huang, H.~Simitci, and S.~Yekhanin, ``{On the Locality of
  Codeword Symbols},'' \emph{IEEE Transactions on Information Theory}, vol.~58,
  no.~11, pp. 6925--6934, Nov. 2012.

\bibitem{Kamath2014}
G.~M. Kamath, N.~Prakash, V.~Lalitha, and P.~V. Kumar, ``{Codes With Local
  Regeneration and Erasure Correction},'' \emph{IEEE Transactions on
  Information Theory}, vol.~60, no.~8, pp. 4637--4660, Aug. 2014.

\bibitem{Tamo2014}
I.~Tamo and A.~Barg, ``{A Family of Optimal Locally Recoverable Codes},''
  \emph{IEEE Trans. Inf. Theory}, vol.~60, no.~8, pp. 4661--4676, 2014.

\bibitem{Silberstein2015}
N.~Silberstein, A.~S. Rawat, and S.~Vishwanath, ``{Error-Correcting
  Regenerating and Locally Repairable Codes via Rank-Metric Codes},''
  \emph{IEEE Transactions on Information Theory}, vol.~61, no.~11, pp.
  5765--5778, Nov. 2015.

\bibitem{Holzbaur2018}
L.~Holzbaur and A.~Wachter-Zeh, ``List decoding of locally repairable codes,''
  in \emph{2018 IEEE International Symposium on Information Theory
  (ISIT)}.\hskip 1em plus 0.5em minus 0.4em\relax IEEE, 2018, pp. 1331--1335.

\bibitem{Rashmi2012}
K.~V. Rashmi, N.~B. Shah, K.~Ramchandran, and P.~V. Kumar, ``{Regenerating
  codes for errors and erasures in distributed storage},'' in \emph{2012 IEEE
  International Symposium on Information Theory}.\hskip 1em plus 0.5em minus
  0.4em\relax IEEE, Jul. 2012, pp. 1202--1206.

\bibitem{Tamo2017}
I.~Tamo, M.~Ye, and A.~Barg, ``{Fractional decoding: Error correction from
  partial information},'' \emph{IEEE International Symposium on Information
  Theory - Proceedings}, no. 1030, pp. 998--1002, 2017.

\bibitem{Blaum2013}
M.~Blaum, J.~L. Hafner, and S.~Hetzler, ``{Partial-MDS Codes and Their
  Application to RAID Type of Architectures},'' \emph{IEEE Transactions on
  Information Theory}, vol.~59, no.~7, pp. 4510--4519, Jul. 2013.

\bibitem{Dikaliotis2010}
T.~K. {Dikaliotis}, A.~G. {Dimakis}, and T.~{Ho}, ``Security in distributed
  storage systems by communicating a logarithmic number of bits,'' in
  \emph{2010 IEEE International Symposium on Information Theory}, June 2010,
  pp. 1948--1952.

\bibitem{Gopalan2013}
P.~{Gopalan}, C.~{Huang}, B.~{Jenkins}, and S.~{Yekhanin}, ``Explicit maximally
  recoverable codes with locality,'' \emph{IEEE Transactions on Information
  Theory}, vol.~60, no.~9, pp. 5245--5256, Sep. 2014.

\bibitem{metzner1990general}
J.~J. Metzner and E.~J. Kapturowski, ``{A General Decoding Technique Applicable
  to Replicated File Disagreement Location and Concatenated Code Decoding},''
  \emph{IEEE Trans. Inform. Theory}, vol.~36, no.~4, pp. 911--917, 1990.

\bibitem{tamo2016optimal}
I.~Tamo, D.~S. Papailiopoulos, and A.~G. Dimakis, ``Optimal locally repairable
  des and connections to matroid theory,'' \emph{IEEE Transactions on
  Information Theory}, vol.~62, no.~12, pp. 6661--6671, 2016.

\bibitem{Silberstein2013}
N.~Silberstein, A.~S. Rawat, O.~O. Koyluoglu, and S.~Vishwanath, ``{Optimal
  locally repairable codes via rank-metric codes},'' \emph{IEEE International
  Symposium on Information Theory - Proceedings}, pp. 1819--1823, 2013.

\bibitem{gabrys2018constructions}
R.~Gabrys, E.~Yaakobi, M.~Blaum, and P.~H. Siegel, ``Constructions of partial
  mds codes over small fields,'' \emph{IEEE Transactions on Information
  Theory}, 2018.

\bibitem{Blaum2016}
M.~{Blaum}, J.~S. {Plank}, M.~{Schwartz}, and E.~{Yaakobi}, ``Construction of
  partial mds and sector-disk codes with two global parity symbols,''
  \emph{IEEE Transactions on Information Theory}, vol.~62, no.~5, pp.
  2673--2681, May 2016.

\bibitem{Huang2007}
C.~{Huang}, M.~{Chen}, and J.~{Li}, ``Pyramid codes: Flexible schemes to trade
  space for access efficiency in reliable data storage systems,'' in
  \emph{Sixth IEEE International Symposium on Network Computing and
  Applications (NCA 2007)}, July 2007, pp. 79--86.

\bibitem{Chen2007}
M.~{Chen}, C.~{Huang}, and J.~{Li}, ``On the maximally recoverable property for
  multi-protection group codes,'' in \emph{2007 IEEE International Symposium on
  Information Theory}, June 2007, pp. 486--490.

\bibitem{krachkovsky1997decoding}
V.~Y. Krachkovsky and Y.~X. Lee, ``{Decoding for Iterative Reed--Solomon Coding
  Schemes},'' \emph{IEEE Transactions on Magnetics}, vol.~33, no.~5, pp.
  2740--2742, 1997.

\bibitem{schmidt2009collaborative}
G.~Schmidt, V.~R. Sidorenko, and M.~Bossert, ``{Collaborative Decoding of
  Interleaved Reed--Solomon Codes and Concatenated Code Designs},'' \emph{IEEE
  Transactions on Information Theory}, vol.~55, no.~7, pp. 2991--3012, 2009.

\bibitem{Overbeck_Diss_InterleveadGab}
R.~Overbeck, ``{Public Key Cryptography based on Coding Theory},'' Ph.D.
  dissertation, TU Darmstadt, Darmstadt, Germany, 2007.

\bibitem{migler2004}
T.~Migler, K.~E. Morrison, and M.~Ogle, ``Weight and rank of matrices over
  finite fields,'' \emph{arXiv preprint math/0403314}, 2004.

\end{thebibliography}

\end{document}